\documentclass[11pt, a4paper, reqno]{amsart}

\usepackage{amsmath}
\usepackage{amsthm}

\usepackage{color}
\usepackage[dvips]{graphicx}
\usepackage{algorithmic}

\DeclareGraphicsExtensions{ps, eps, png, jpeg}

\renewcommand{\H}{\mathbb{H}}
\newcommand{\Z}{\mathbb{Z}}
\newcommand{\R}{\mathbb{R}}
\newcommand{\C}{\mathbb{C}}
\newcommand{\tr}{\mathop{\mathrm{tr}}\nolimits}
\newcommand{\on}[2]{\mathop{\null#2}\limits^{#1}}
\newcommand{\oover}[1]{\on{\circ}{#1}}

\newcommand{\curl}{\mathop{\mathrm{curl}}\nolimits}

\newcommand{\norm}[1]{\left\|#1\right\|}
\newcommand{\pair}[1]{\langle#1\rangle}

\newtheorem{thm}{Theorem}

\newtheorem{prop}{Proposition}

\theoremstyle{definition}

\newtheorem*{definition}{Definition}

\theoremstyle{remark}
\newtheorem*{remark}{Remark}

\gdef\am{&}%
\let\bs\cr%
\catcode`\&=13%
\gdef&{\am\displaystyle}%
\catcode`\&=4%

\gdef\endddarray{%
\endarray
\catcode`\&=4%
\let\cr\bs
}

\begin{document}

\title[A new doubly discrete analogue of smoke ring flow]{A new doubly
  discrete analogue of smoke ring flow and the real time simulation of
  fluid flow} 

\author{Ulrich Pinkall}
\address{Institut f\"ur Mathematik, Technische Universit\"at Berlin\\
  Stra{\ss}e des 17. Juni 136\\ 10623 Berlin\\ Germany}
\email{pinkall@math.tu-berlin.de} 
\author{Boris Springborn}
\email{springb@math.tu-berlin.de}
\author{Steffen Wei{\ss}mann}
\email{weissman@math.tu-berlin.de} \thanks{The authors are supported
  by DFG Research Center \textsc{Matheon}.}

\begin{abstract}
  Modelling incompressible ideal fluids as a finite collection of
  vortex filaments is important in physics (super-fluidity, models for
  the onset of turbulence) as well as for numerical algorithms used in
  computer graphics for the real time simulation of smoke.  Here we
  introduce a time-discrete evolution equation for arbitrary closed
  polygons in 3-space that is a discretisation of the localised
  induction approximation of filament motion. This discretisation
  shares with its continuum limit the property that it is a completely
  integrable system. We apply this polygon evolution to a significant
  improvement of the numerical algorithms used in Computer Graphics.
\end{abstract}

\maketitle

\section{Introduction}
\label{intro}

The motion of vortex filaments in an incompressible, inviscid fluid has aroused
considerable interest in quite different areas:

\medskip\noindent \textbf{Differential geometry.} The limiting case of 
infinitely thin vortex filaments leads to an evolution equation for closed 
space curves $\gamma$,
\begin{equation}
\label{smokeringflow}
\dot \gamma = \gamma' \times \gamma''.
\end{equation}
Equation~(\ref{smokeringflow}) was discovered in the beginning of the 20th 
century by Levi-Civita and his student Da Rios \cite{DaRios} and is called the 
\emph{smoke ring flow} or \emph{localised induction approximation}. In 1972 
Hasimoto \cite{Hasimoto} discovered that the smoke ring flow is in fact a 
completely integrable Hamiltonian system equivalent to the non-linear 
Schr{\"o}dinger equation. See~\cite{Ricca} for more details on the history of 
the smoke ring equation. Subsequently the smoke ring flow has been studied by 
differential geometers as a natural evolution equation for space curves 
\cite{CaliniIvey,CieslinskiGragertSym,Ivey,langerPerline}. 
Also discrete versions of the smoke ring flow in the form of completely 
integrable evolution equations for polygons with fixed edge length have been 
developed \cite{tim,timPaper,DoliwaSantini}.

\medskip\noindent
\textbf{Fluid dynamics.}
As will be explained below, for applications in fluid mechanics a finite 
thickness of the vortex filaments has to be taken into account. The transition 
from infinitely thin filaments to filaments of finite thickness involves the 
incorporation of long range interactions (governed by the Biot-Savart law) 
between different filaments and different parts of the same filament into the 
purely local evolution equation~(\ref{smokeringflow}). The resulting evolution 
of vortex filaments has been extensively studied both numerically and in 
the context of explaining the onset of turbulence \cite{Chorin_book}. Including 
in addition a small amount of viscosity in the equations leads to striking 
physical effects like vortex reconnection 
\cite{Koplik_Levine,Kivotides_Leonard,Chatelain_Kivotides_Leonard} and 
numerical techniques like ``hairpin removal'' \cite{Chorin_1,Chorin_2}.

\medskip\noindent 
\textbf{Computer graphics.} Filament-based methods
for fluid simulation are becoming important in Computer Graphics for
special effects in motion pictures and for real time applications like
computer games \cite{AN05,ANSN06}. Here the emphasis is on physical
correctness and speed rather than numerical accuracy. Filament methods
are ideal for these applications because complicated fluid motions can
be created by a graphics designer by modelling the initial positions
and strengths of the filaments. Moreover, filament methods work in
unbounded space rather than in a bounded box (as is the case for
grid-based methods \cite{Stam_Stable_fluids}). This is desirable for
the simulation of smoke.

\medskip

The main goal of this paper is to improve the numerical algorithms currently 
used in Computer Graphics by applying the recent knowledge from Discrete 
Differential Geometry to the motion of polygonal smoke rings. Our
method makes it possible to model thin filaments by polygons with 
arbitrarily few vertices. For comparison, using current methods to
model a circular smoke ring which is thin enough to entrain smoke in a
torus shape, it necessary to use a regular polygon with at least
$800$ vertices. 

In Section \ref{filaments} we will explain the evolution equation for
systems of vortex filaments that we will discretise. The resulting
equations of motion are still Hamiltonian like the smoke ring
flow~(\ref{smokeringflow}). However, since already Poincar\'e knew
that a system of vortex filaments consisting of more than three
parallel straight lines (the ``$n$-vortex problem'') fails to be an
integrable system~\cite[p.~58f]{arnold}, we do not believe that this
system is an integrable Hamiltonian system. Nevertheless it is a small
perturbation of the integrable system constituted by the limit of
infinitely thin filaments. This might be interesting for future
investigations along the lines of KAM theory.

In Section~\ref{polygonalFilaments} we consider polygonal vortex filaments. In
this case, there is an elementary formula (\ref{blurredAntiDerivative}) for the Biot-Savart integral.

In Section~\ref{darboux} we will develop an extension of the known
discrete-time smoke ring flow for polygons of constant edge lengths to
arbitrary polygons. This is needed because after including the long
range Biot-Savart interactions, the lengths of the edges will be no
longer constant in time. 

In the theory of integrable systems it is known at least since the 1980s that 
integrable difference equations may be interpreted as Darboux transformations of 
integrable differential equations~\cite{LeviBenguria,Levi,NijhoffQuispelCapel}. 
In the meantime, this seminal discovery has lead to a reversed point of view, 
where the discrete integrable systems are considered fundamental and the 
continuous systems appear as smooth limits (see for 
example~\cite{Adler_Bobenko_Suris} and the references therein). In this vein, 
we will in Section~\ref{darboux} define the discrete-time integrable system in 
terms of iterated Darboux transformations of polygons and show afterwards that 
the smoke ring flow is obtained as a smooth limit.
 
In Section~\ref{simulation} we will describe our numerical method that
very efficiently models the motion of fluids near the smoke ring
limit.

\section{Euler's Equation for Vortex Filaments}
\label{filaments}
Consider an incompressible, inviscid fluid in euclidean 3-space whose velocity 
field $u$ vanishes at infinity and whose vorticity $\omega = \curl u$ is 
compactly supported. Then $u$ can be reconstructed from $\omega$ by the 
\emph{Biot-Savart} formula
\begin{equation}
\label{biot-savart}
u(x) = -\frac{1}{4 \pi} \int_{\R^3} \frac{x-z}{\norm{x-z}^3} \times \omega(z) 
\, dz.
\end{equation}
The equation of motion can then be written as
\begin{equation}
\label{eq:eq_of_motion}
\dot \omega = [\omega, u].
\end{equation}
Viewed as an evolution equation on the vector space $\mathcal{M}$ of compactly 
supported divergence-free vector fields on $\mathbb{R}^3$ this is a Hamiltonian 
system: A symplectic form~$\sigma$ on $\mathcal{M}$ is defined as follows. Let 
$\omega\in\mathcal{M}$ and $\dot \omega, \oover{\omega}\in 
T_{\omega}\mathcal{M}$. Then
\begin{equation}
\sigma_\omega (\dot \omega, \oover{\omega}) = \int_{\mathbb{R}^3} \det(\omega, 
\dot \omega, \oover{\omega}).
\end{equation}
Let $H:\mathcal{M}\rightarrow\R$ be the quadratic function
\begin{equation}
H = \int\!\!\!\int \frac{\pair{\omega(x), \omega(y)}}{\norm{x - y}} \, dx\, dy,
\end{equation}
where $\pair{\cdot,\cdot}$ is the euclidean scalar product on $\R^3$. Then $H$ 
is the Hamiltonian for the dynamical system~(\ref{eq:eq_of_motion}). See 
\cite{arnold, EaM} for more details on this Hamiltonian description of 
ideal fluids.

If the vorticity of a fluid is concentrated on a closed curve $\gamma$ in a 
delta-function like manner, by Equation~(\ref{biot-savart}) the resulting 
velocity field $u$ becomes

\begin{equation}
\label{curve-biot-savart}
u(x) = -\frac{\Gamma}{4 \pi} \oint \frac{x-\gamma(s)}{\|x-\gamma(s)\|^3} \times 
\gamma'(s) \, ds.
\end{equation}
Here $\Gamma$ is the circulation around the filament. The problem with
Equation~(\ref{curve-biot-savart}) is that in order to determine the motion of $\gamma$
itself, $u$ has to be evaluated on $\gamma$, which results in a logarithmically
divergent integral.
Usually, this problem is handled by considering a vorticity field concentrated 
in a tube around $\gamma$ of small but finite radius $r$. For small $r$ the 
velocity in this tube is dominated by a term proportional to the localised 
induction approximation. (See, for example,~\cite[p.~36f]{Saffman}.) Here we 
want to derive the smoke ring flow by taking the limit $r\rightarrow 0$. In 
order to prevent vortex filaments acquiring infinite speed, one has to scale 
the circulation $\Gamma$ down to zero when performing the limit to infinitely 
thin filaments. This means that the fluid velocity (\ref{curve-biot-savart}) 
goes to zero as well.

\begin{figure}
\centering\includegraphics[width=0.5\textwidth]{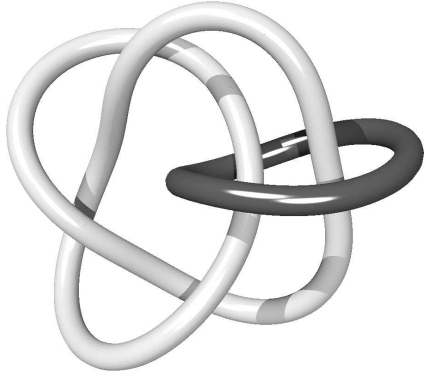}
\caption{The space of links as the phase space for vortex filaments.}
\label{fig:link}
\end{figure}

The resulting picture is then as follows: The fluid is completely at rest away 
from the filaments while the filaments just cut through the fluid with finite
speed according to the smoke ring flow:

\begin{equation}
\label{smokeringflow_links}
\dot \gamma_j = K_j \, \gamma_j' \times \gamma_j''.
\end{equation}
Here the constants $K_j$ account for the fact that the circulation of the
different filaments $\gamma_j$ might go to zero at a different rate.

Equation~(\ref{smokeringflow_links}) can be viewed as a completely integrable Hamiltonian
system on the space of weighted links (see Figure~\ref{fig:link}) endowed with the symplectic form

\begin{equation}
\label{symplectic_links}
\sigma_\gamma (\dot \gamma, \oover{\gamma}) = \sum_j K_j \oint_{\gamma_j}
\det(\gamma_j', \dot \gamma, \oover{\gamma}).
\end{equation}
For single curves this symplectic form is due to V.~I. Arnold \cite{arnold}.
The corresponding Hamiltonian is a weighted sum of the filament lengths
\begin{equation*}
H = \sum_j K_j \ \mathrm{Length}(\gamma_j).
\end{equation*}
Equation~(\ref{smokeringflow_links}) can be obtained (using a simple renormalisation of 
time) as a limit as $a \rightarrow 0$ of the following system: Stick with
(\ref{symplectic_links}) as the symplectic form, with $K_j$ replaced by the
non-zero circulation $\Gamma_j$ around $\gamma_j$. As a Hamiltonian, use 
\begin{equation*}
H = \sum_{i, j} \frac{\Gamma_i \Gamma_j}{8 \pi} \oint\!\!\oint 
\frac{\pair{\gamma_i'(s), \gamma_j'(\tilde s)}}{\sqrt{a^2 + \norm{\gamma_i(s) - 
\gamma_j(\tilde s)}^2}} \, ds \, d \tilde s.
\end{equation*}
The resulting equation of motion is
\begin{equation}
\label{filamentEquation}
\dot \gamma_i(s) = - \sum_j \frac {\Gamma_j}{4\pi} \oint 
\frac{\gamma_i(s) -\gamma_j(\tilde s)}%
{\sqrt{a^2+|\gamma_i(s)-\gamma_j(\tilde s)|^2}^{\,3}} \times \gamma_j'(\tilde s)
 \, d\tilde s.
\end{equation}
This equation of motion~(\ref{filamentEquation}) can also be derived
as follows:
\begin{itemize}
  \item Smooth the delta-function like vorticity field $\omega_0$ of the link by
  a suitable convolution kernel and obtain
\begin{equation*}
\omega(x) = \frac{3 a^2}{4\pi}
\int_{\R^3}\frac{\omega_0(y)}{\sqrt{a^2+|x-y|^2}^{\,5}} \, dy.
\end{equation*}
\item Compute the corresponding velocity field $u$ with $\curl u = \omega$:
\begin{equation}
\label{velocityField}
u(x) = - \frac {\Gamma}{4\pi} \sum_j \oint 
\frac{x-\gamma_j(s)}{\sqrt{a^2+|x-\gamma_j(s)|^2}^{\,3}} \times \gamma_j'(s) \, ds.
\end{equation}
\item Evaluate $u$ on the filaments to obtain (\ref{filamentEquation}).
\end{itemize}
To summarise: We model fluid motion near the filament limit by a
Hamiltonian system on the space of weighted links. This system is
still Hamiltonian but no longer integrable. Nevertheless it still has
all the constants of motion induced by invariance with respect to the
euclidean symmetry group.  For example the weighted sum of the area
vectors
\begin{equation*}
  A = \sum_j \Gamma_j \oint \gamma_j' \times \gamma_j
\end{equation*}
is one of the preserved quantities. (Compare
Theorem~\ref{areaInvariant} of Section~\ref{darboux}.)

The physical approximation implicit in our model is that we ignore possible 
deformations of the internal structure of the filaments and reduce everything 
to the evolution of the filament curves. The finite thickness of the filaments 
is taken into account by applying a fixed convolution kernel.

\section{Polygonal Vortex Filaments}
\label{polygonalFilaments}

In order to develop a numerical method for modelling fluid motion near the 
filament limit we have to discretise the vortex filaments, {i.e.}~we 
replace them by 
polygons. If $\gamma$ is a piecewise linear parametrisation of a closed 
polygon, on each edge we have $\gamma''=0$ and we find an explicit 
anti-derivative for the integrands of equation~(\ref{velocityField}):
\begin{equation}
\label{blurredAntiDerivative}
\Bigg( \frac {\pair {\gamma, \gamma'}}{\sqrt {a^2+|\gamma|^2}\,\, 
(|\gamma'|^2 a^2+|\gamma \times \gamma'|^2)} \gamma \times \gamma' \Bigg)' 
= \frac{\gamma \times \gamma'}{\sqrt{a^2+|\gamma|^2}^{\,3}}.
\end{equation}
Here we have abbreviated $x - \gamma_j(s)$ to $\gamma$, $\gamma_j'(s)$ to
$\gamma'$ and the prime is derivation with respect to $s$.

Inspection of Equation~(\ref{blurredAntiDerivative}) reveals the following problem:
The two adjacent edges have no influence at all on the velocity of a vertex.
This amounts to effectively employing a distance cut-off in order to regularise
the singular integral (\ref{curve-biot-savart}) for points on $\gamma$. It is
known \cite{Saffman} that this is roughly equivalent to modelling vortex tubes
of thickness equal to the edge length of the polygon. Using this model we
would therefore be unable to model thin (and therefore fast) filaments without
using excessively many edges for each polygon.

The contribution of local effects behaves like the smoke ring flow and the 
resulting equation of motion for a vertex $\gamma_i$ of a polygonal vortex 
filament $\gamma$ is then
\begin{equation}
\label{eq:vertex_motion}
\dot\gamma_i = u(\gamma_i) + \lambda \kappa_i b_i,
\end{equation}
where $u$ is given by Equation~(\ref{velocityField}) using
(\ref{blurredAntiDerivative}), $\kappa_i b_i$ denotes
curvature times binormal at $\gamma_i$, and $\lambda$ is constant for
fixed $a$.  Since the non-local effects quickly destroy any arc-length
parametrisation (i.e.~the lengths of the different edges of the
polygon) and we do not have an adequate notion of curvature for
arbitrary polygons, we can not evaluate~(\ref{eq:vertex_motion})
directly.

On the other hand, for polygons with constant edge lengths it is known that 
the doubly discrete smoke ring (or Hasimoto) flow~\cite{timPaper} captures 
excellently the qualitative behaviour of the smooth smoke ring flow. In the 
next section we will discuss a version of this doubly 
discrete smoke ring flow which works also for polygons with varying edge 
lengths.

\section{Darboux Transformation of Polygons}
\label{darboux}

In this section we develop a discrete-time evolution for closed polygons that 
has the smoke ring flow (\ref{smokeringflow}) as a limit when the polygon 
approaches a smooth curve and the time-step goes to zero. This evolution 
(obtained by iterating so-called Darboux transformations) shares with its 
continuum limit the property that it is a completely integrable system in the 
sense that it comes from a Lax pair
of quaternionic $2 \times 2$-matrices with a 
spectral parameter. (This system therefore fits into the framework of~\cite{Bobenko_Suris}.) The
constants of the motion of the discrete
system converge to constants of the motion of the smooth system in the limit.

 Let $\gamma:\Z\rightarrow\R^3$ be an immersed polygon in 
 $\R^3$, where \emph{immersed} means that $\gamma_i \neq \gamma_{i+1}$ for all 
 $i \in \Z$, and let $S_i=\gamma_{i+1}-\gamma_i$. If $\gamma$ is periodic with some
 period $n$, then the polygon is \emph{closed} and $\gamma$ may be interpreted as a
 function on $\Z/n\Z$. In the following, we identify $\R^3$
 with the imaginary quaternions $\mathrm{Im}\,\H=\{xi+yj+zk\,|\,x,y,z\in\R\}$.
 \begin{definition}
 A polygon $\eta$ is called a \emph{Darboux transform} of $\gamma$ with twist 
 parameter $r \in \R$ and distance $l > 0$, if $\|\eta_i - \gamma_i\| = l$ for 
 all $i \in \Z$, and the normalised difference vectors~$T_i$ defined by $ 
 l T_i = \eta_i - \gamma_i $ satisfy the quaternionic equation
\begin{equation}
\label{d2}
 T_{i+1} = (-r + l T_i - S_i) T_i (-r + l T_i - S_i)^{-1}.
\end{equation}
\end{definition}

The Darboux transformation of polygons and its relationship with the
nonlinear Schr\"odinger equation and smoke ring flow was treated
in~\cite{timPaper} under the assumption that the polygon $\gamma$ has
constant edge length.  To drop this assumption was suggested to us by
Tim Hoffmann \cite{timOral}.

Geometrically, Equation~(\ref{d2}) has the following meaning (see 
Figure~\ref{fig:darboux}).
\begin{figure}
\centering
\begin{picture}(0,0)%
\includegraphics{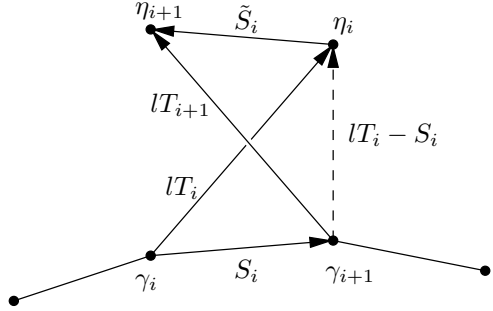}%
\end{picture}%
\setlength{\unitlength}{4144sp}%
\begingroup\makeatletter\ifx\SetFigFontNFSS\undefined%
\gdef\SetFigFontNFSS#1#2#3#4#5{%
  \reset@font\fontsize{#1}{#2pt}%
  \fontfamily{#3}\fontseries{#4}\fontshape{#5}%
  \selectfont}%
\fi\endgroup%
\begin{picture}(2860,1892)(1046,-1626)
\put(2971, 29){\makebox(0,0)[lb]{\smash{{\SetFigFontNFSS{10}{12.0}{\familydefault}{\mddefault}{\updefault}{\color[rgb]{0,0,0}$\eta_i$}%
}}}}
\put(2926,-1456){\makebox(0,0)[lb]{\smash{{\SetFigFontNFSS{10}{12.0}{\familydefault}{\mddefault}{\updefault}{\color[rgb]{0,0,0}$\gamma_{i+1}$}%
}}}}
\put(3061,-646){\makebox(0,0)[lb]{\smash{{\SetFigFontNFSS{10}{12.0}{\familydefault}{\mddefault}{\updefault}{\color[rgb]{0,0,0}$lT_i-S_i$}%
}}}}
\put(1801,119){\makebox(0,0)[lb]{\smash{{\SetFigFontNFSS{10}{12.0}{\familydefault}{\mddefault}{\updefault}{\color[rgb]{0,0,0}$\eta_{i+1}$}%
}}}}
\put(1891,-466){\makebox(0,0)[lb]{\smash{{\SetFigFontNFSS{10}{12.0}{\familydefault}{\mddefault}{\updefault}{\color[rgb]{0,0,0}$lT_{i+1}$}%
}}}}
\put(1981,-961){\makebox(0,0)[lb]{\smash{{\SetFigFontNFSS{10}{12.0}{\familydefault}{\mddefault}{\updefault}{\color[rgb]{0,0,0}$lT_i$}%
}}}}
\put(2386,-1456){\makebox(0,0)[lb]{\smash{{\SetFigFontNFSS{10}{12.0}{\familydefault}{\mddefault}{\updefault}{\color[rgb]{0,0,0}$S_i$}%
}}}}
\put(1801,-1501){\makebox(0,0)[lb]{\smash{{\SetFigFontNFSS{10}{12.0}{\familydefault}{\mddefault}{\updefault}{\color[rgb]{0,0,0}$\gamma_i$}%
}}}}
\put(2386, 29){\makebox(0,0)[lb]{\smash{{\SetFigFontNFSS{10}{12.0}{\familydefault}{\mddefault}{\updefault}{\color[rgb]{0,0,0}$\tilde S_i$}%
}}}}
\end{picture}%
	\caption{A polygon $\gamma$ and an edge of its Darboux transform $\eta$.}
	\label{fig:darboux}
\end{figure}
The difference vector $T_{i+1}$ is obtained from $T_i$ by a rotation with axis 
$lT_i-S_i$. The quadrilateral $\gamma_i\gamma_{i+1}\eta_{i+1}\eta_i$ is 
therefore a ``folded parallelogram''. In particular, corresponding edges of
$\gamma$ and $\eta$ have the same length. The
angle of rotation is $2\arctan({\|lT_i-S_i\|}/{r})$. For $r=0$ it is $\pi$. For 
$r\rightarrow\pm\infty$, it goes to zero and in the limit the Darboux
transformation becomes a translation. 

Equation~(\ref{d2}) can be written in the form 
\begin{equation}
T_{i+1} = (aT_i+b)(cT_i+d)^{-1},
\end{equation}
where $a,b,c,d\in\H$ depend on $S_i$ and the parameters $l,r$. That is, for 
each $i \in \Z$, $T_{i+1}$ is obtained by applying a quaternionic 
fractional linear transformation $f_i:\bar\H\rightarrow\bar\H$ to $T_i$, where 
$\bar\H=\H\cup\{\infty\}$. Indeed, 
(\ref{d2}) is equivalent to
\begin{equation}
\label{d6}
T_{i+1} = \big(l T_i - r - S_i\big) \big((r + S_i) T_i + l\big)^{-1}.
\end{equation}
To see this note that $T_i^{-1} = -T_i$ because $T_i$ is a purely 
imaginary unit quaternion, and hence $T_i (-r + l T_i - S_i)^{-1}=(rT_i+l+S_iT_i)$.

It is convenient to rewrite fractional linear transformations as matrix 
multiplication. Just as the extended complex plane $\bar\C=\C\cup\{\infty\}$ 
can be identified with the Riemann sphere $S^2$ and with the complex projective
line $\C\mathrm{P}^1$, $\bar\H\cong S^4\cong\H\mathrm{P}^1$. The quaternionic 
projective line $\H\mathrm{P}^1$ is the set of (quaternionic) 
$1$-dimensional subspaces of the vector space $\H^2$ over $\H$. We consider 
$\H^2$ as right vector space: the product of a vector $\big({ p\atop 
q}\big)\in\H^2$ and a scalar $\lambda\in\H$ is $\big({p\atop q}\big)\lambda = 
({p\lambda \atop q\lambda})$. A point 
\begin{equation*}
	\Bigg[{p\atop q}\Bigg]=\Bigg({p\atop q}\Bigg)\H\in\H\mathrm{P}^1
\end{equation*}
corresponds to the point $pq^{-1}\in\bar\H$, and $p,q$ are 
quaternionic homogeneous coordinates for this point. Now any fractional linear
transformation 
of $\bar\H$ can be written as quaternionic $2\times2$-matrix acting from the 
left on quaternionic homogeneous coordinates of $\H\mathrm{P}^1$: Writing $T_i$
in quaternionic homogeneous coordinates, 
\begin{equation*}
T_{i} = T^{(1)}_{i} (T^{(2)}_i)^{-1},
\end{equation*}
one obtains from (\ref{d6})
\begin{equation}
	\label{eq:Ui}
	\left(\begin{array}{c} 
    	T^{(1)}_{i+1} \\ 
    	T^{(2)}_{i+1}
	\end{array}\right) 
	= 
	U_i(l,r)
	\left(\begin{array}{c} 
      	T^{(1)}_i \\ 
      	T^{(2)}_i  
	\end{array}\right),\quad
	U_i(\lambda,\rho) := 
	\left(
	\begin{array}{cc} 
    	\lambda & -\rho -S_i \\ 
        \rho+S_i & \lambda 
    \end{array}\right).
\end{equation}

The following Theorem~\ref{thm:lax} characterises the Darboux transformations of polygons via a Lax
pair of quaternionic $2\times 2$-matrices with spectral parameter.
Theorem~\ref{thm:permutability} is a permutability theorem for these Darboux
transformations. 

\begin{thm}[Lax pair]
	\label{thm:lax}
	Let $S_i=\gamma_{i+1}-\gamma_i$, $|T_i|=1$, and let $U_i(\lambda,\rho)$ be defined
	by~(\ref{eq:Ui}) and
	\begin{equation*}
    	\tilde{U}_i(\lambda,\rho)=
    	\left(
    	\begin{array}{cc} 
        	\lambda & -\rho -\tilde S_i \\ \rho+\tilde S_i & \lambda 
        \end{array}
		\right),
	\end{equation*}%
	\begin{equation*}
		V_i(\lambda,\rho)=
		\left(
    	\begin{array}{cc}
      		\lambda &  - \rho + r - l T_i\\
      		\rho - r + l T_i & \lambda
        \end{array}
		\right).
    \end{equation*}
	Then 
	\begin{equation}
    	\label{eq:UV_compat}
    	V_{i+1}(\lambda,\rho)U_i(\lambda,\rho)
    	=\tilde U_i(\lambda,\rho)V_i(\lambda,\rho)
    \end{equation}
	for all $\lambda,\rho\in\R$, if and only if $S$ and $T$ satisfy~(\ref{d2}) and	
	\begin{equation}
    	\label{eq:STclosure}
      	l T_{i+1} + S_i=\tilde S_i + l T_i.
    \end{equation}
	That is, if and only if $\eta=\gamma+lT$ is a Darboux transform of $\gamma$
	with twist parameter $r$ and distance $l$, and $\tilde S_i=\eta_{i+1}-\eta_i$.
\end{thm}

\noindent
Of course (\ref{eq:UV_compat}) means that the following diagram
commutes: 
\begin{equation*}
	\begin{array}{ccc}
    	\H^2 & \stackrel{\mbox{$\tilde U_i$}}{\longrightarrow}  & \H^2 \\
    	\makebox[0pt][r]{$V_i$}\Big\uparrow & 
    	& \Big\uparrow\makebox[0pt][l]{$V_{i+1}$}\\ 
    	\H^2 & \stackrel{\mbox{$U_i$}}{\longrightarrow}  & \H^2 
	\end{array}
\end{equation*}

\begin{proof}
Note that in general for quaternionic $2 \times 2$-matrices with
$a_+, a, b, \tilde b \in \H$
and $\lambda \in \R$ the equality
\begin{equation*}
\left(\begin{array}{cc} \lambda & a_+ \\ -a_+ & \lambda \end{array} \right)
\left(\begin{array}{cc} \lambda & b \\ -b & \lambda \end{array} \right)
=
\left(\begin{array}{cc} \lambda & \tilde{b} \\ -\tilde{b} & \lambda \end{array} \right)
\left(\begin{array}{cc} \lambda & a \\ -a & \lambda \end{array} \right)
\end{equation*}
is equivalent to
\begin{equation*}
a_+ b = \tilde{b} a
\qquad \textrm{and} \qquad
\lambda (a_+ + b) = \lambda (\tilde b + a).
\end{equation*}
It follows that~(\ref{eq:UV_compat}) holds for all $\lambda\in\R$, if and only
if~(\ref{eq:STclosure}) holds and
\begin{equation*}
(-\rho+r-lT_{i+1})(-\rho-S_i)=(-\rho-\tilde S_i)(-\rho+r-lT_{i}).
\end{equation*}
Use~(\ref{eq:STclosure}) to eliminate $\tilde S_i$ from this equation and
gather terms of equal power in $\rho$ on both sides. The coefficients of
$\rho^2$ are both 1, and the coefficients of $\rho$ are obviously equal. What
remains is the equation
\begin{equation*}
(r-lT_{i+1})(-S_i)=(-S_i-lT_{i+1}+lT_i)(r-lT_i).
\end{equation*}
Solve for $T_{i+1}$ to obtain~(\ref{d2}).
\end{proof}

\begin{thm}[Permutability]
	\label{thm:permutability}
	Suppose $\eta=\gamma+lT$ is a Darboux transform of $\gamma$ with twist
	parameter~$r$ and distance~$l$, and $\hat\eta=\gamma+\lambda\hat T$ is a
	Darboux transform of $\gamma$ with twist parameter~$\rho$ and distance
	$\lambda$, then $\eta + \lambda\tilde T$ with
	\begin{equation}  
	    \label{eq:tilde_T}
 		\tilde T = \big(\lambda \hat{T} - \rho + r - l T\big)
 			\big((\rho - r + l T)\hat T + \lambda\big)^{-1}
	\end{equation}	
	is a Darboux transformation of $\eta$ with twist parameter $\rho$ and distance
	$\lambda$.
\end{thm}

\begin{proof}
	Note that $\tilde T_i$ is obtained by applying the quaternionic fractional
	linear transformation represented by the matrix $V_i(\lambda,\rho)$ to $\hat
	T_i$. Let us write $\tilde T_i=V_i(\lambda,\rho)\hat T_i$
	for short. Since $\hat\eta$ is a Darboux transform of $\gamma$ with twist
	parameter $\rho$ and distance $\lambda$, Equation~(\ref{eq:Ui}) says that 
	$\hat T_{i+1}=U_i(\lambda,\rho)\hat T_i$. 
	Now Theorem~\ref{thm:lax} implies 
	$\tilde T_{i+1}=\tilde U_i(\lambda,\rho)\tilde T_i$ and hence (again by
	Equation~(\ref{eq:Ui})), $\eta + \lambda\tilde T$ is a  Darboux transformation 
	of $\eta$ with twist parameter $\rho$ and distance $\lambda$. 
\end{proof}

Even if $\gamma$ is a closed curve, the curves obtained by
iterating~(\ref{d2}) will in general not close up. However, we will
see that any closed curve has generically two closed Darboux transforms.

The fractional linear transformations $f_i:T_i\mapsto T_{i+1}$ that are
represented by the matrices $U_i(l,r)$ have the special property that they
map the unit sphere $S^2 =\{q\in\mathrm{Im}\,\H\,|\,q^2=-1\}$ to
itself. This follows directly 
from~(\ref{d2}). Hence the restrictions $f_i|_{S^2}$ are
M\"obius transformations of $S^2$. 
In fact, they are orientation preserving M\"obius transformations: By
continuity, it is enough to check this for a particular value of $r$
and~$l$; and for $r=0$, $l=0$ one obtains $T_{i+1}=S_iT_iS_i^{-1}$, which
is a $180^\circ$ rotation with axis~$S_i$.

In order to find for given $l$, $r$ the 
closed Darboux transforms of $\gamma$, one has to look for choices of the 
initial unit vector $T_0$ such that the recursion (\ref{d6}) generates a 
sequence with period $n$, i.e.~$T_0 = T_n$. The composition 
$f_{n-1}\circ\ldots\circ f_0$, which maps $T_0 \mapsto T_n$, is represented by
the monodromy matrix
\begin{equation*}
	H_{l,r}=U_{n-1}(l,r)\cdots U_{2}(l,r) U_{1}(l,r) U_{0}(l,r).
\end{equation*}
It is is itself an 
orientation-preserving M\"obius transformation of the unit sphere $S^2$ onto 
itself. For special cases (we will see below that this cannot happen for all 
$l$, $r$) this M\"obius-transformation could be the identity, but in general it 
will have exactly two fixed points (counted with multiplicity).

With each closed curve $\gamma$ we have thus associated a monodromy map 
$f_{n-1}\circ\ldots\circ f_0$. 
$T_0$ will be a fixed point 
if and only if 
$\left(T_0 \atop 1 \right)$ is an eigenvector of the monodromy matrix~$H_{l,r}$.
The following theorem is an immediate consequence of Theorem~\ref{thm:lax}.

\begin{thm}
\label{dt1}
Suppose $\eta = \gamma + l T$ is a closed Darboux transform of
$\gamma$ with distance $l$ and twist parameter $r$.  Then for all
$\lambda$ and $\rho$, the monodromy matrix~$H^{\eta}_{\lambda, \rho}$
of $\eta$ is conjugate to the monodromy matrix $H_{\lambda,
  \rho}$ of $\gamma$:
\begin{equation}
  \label{d10}
  H^{\eta}_{\lambda, \rho} = 
  V_{0}(\lambda,\rho) H_{\lambda, \rho} V_{0}(\lambda,\rho)^{-1}.
\end{equation}
This means that if $\left(\hat{T}_0 \atop 1 \right)$ is
an eigenvector of $H_{\lambda, \rho}$, then
$V_{0}(\lambda,\rho)\left(\hat{T}_0 \atop 1 \right)$
is an eigenvector of~$H^{\eta}_{\lambda, \rho}$.
\end{thm}

Moreover, one can compute all closed Darboux transforms of $\eta$ without having
to solve an eigenvalue problem, even without iterating the $f_i$.
Indeed, by Theorem~\ref{thm:permutability},
all closed Darboux transforms of $\eta$ are given by (\ref{eq:tilde_T}).

Theorem~\ref{dt1} implies that apart from the edge lengths there are
many other quantities connected with closed polygons that are
invariant under Darboux transforms: For each $\lambda$, $\rho$ the
conjugacy class of the monodromy matrix $H_{\lambda, \rho}$ is
invariant. We will show that this implies a nice geometric invariant: The area
vector of a closed polygon turns out to be invariant under Darboux
transformations (Theorem~\ref{areaInvariant}). 

To derive the invariance of the area vector from the invariance of the conjugacy
class of the monodromy
matrix, we equip the set of quaternionic $2\times2$-matrices of the form
\begin{equation} 
\label{d17}
\left( \begin{array}{cc} a & -b \\ b & a \end{array} \right), \ \ a, b \in \H
\end{equation}
with the structure of a $\C$-algebra that is isomorphic to $gl(2, \C)$. First
note that a quaternionic $2\times2$-matrix is of the form~(\ref{d17}) precisely if
it commutes with
\begin{equation*}
J = \left( \begin{array}{cc} 0 & -1 \\ 1 & 0 \end{array} \right).
\end{equation*}
Define the multiplication of such a matrix with a scalar $\lambda + i\rho\in\C$ by
\begin{equation}
	\label{eq:complex_multiplication}
(\lambda + i\rho) A = (\lambda I + \rho J) A,
\end{equation}
where $I$ is the identity matrix.

The complex multiples of the identity are then
\begin{equation}
\label{d19}
Z =  (\lambda + i\rho)I = \lambda I + \rho J
= 
\left(
\begin{array}{cc}
	\lambda & -\rho \\ \rho & \lambda
\end{array}
\right).
\end{equation}
Thus we can write $U_i(\lambda,\rho)$ and $V_i(\lambda,\rho)$ as
\begin{equation*}
	U_i(\lambda,\rho)=(\lambda + i \rho) I + J
	\left(
	\begin{array}{cc}
    	S_i & 0 \\ 0 & S_i
    \end{array}
	\right),
\end{equation*}
\begin{equation*}
	V_i(\lambda,\rho)=(\lambda + i \rho) I + J
	\left(
	\begin{array}{cc}
    	-r +l T_i & 0 \\ 0 & -r + l T_i
    \end{array}
	\right).
\end{equation*}

\begin{remark}
This means we can combine $\lambda$ and $\rho$ into one complex spectral
parameter $\lambda + i \rho$.
\end{remark}

Equation~(\ref{d19}) also implies that the trace-free complex matrices in $gl(2,
\C)$ correspond to those
matrices of the form~(\ref{d17}) with $a, b \in \mathrm{Im}\H$. Further, a
matrix of the form~(\ref{d17}) has $a, b \in \mathrm{Im}\H$ precisely if its
square is a matrix of the form~(\ref{d19}), that is, a (complex) multiple
(with multiplication defied by~(\ref{eq:complex_multiplication})) of the
identity. Identifying $\C$ with the matrices of the form (\ref{d19}) we obtain
\begin{equation*}
\frac{1}{2} \tr_{\C} \left( \begin{array}{cc} a & -b \\ b & a \end{array} \right)
= \mathrm{Re}\,a + (\mathrm{Re}\,b) J
\end{equation*}
and
\begin{equation*}
\det_{\C}\nolimits \left( \begin{array}{cc} a & -b \\ b & a \end{array} \right)
= \frac{1}{2}((\tr A)^2 - \tr A^2) = |a^2| -|b^2| + 2 \pair{a,b} J.
\end{equation*}
In particular
\begin{equation*}
\det_{\C}\nolimits \left( \begin{array}{cc} l & -r-S \\ r+S & l \end{array} \right)
= l^2 - r^2 -|S|^2 + 2 l r J,
\end{equation*}
which vanishes precisely when $r=0, l=\pm |S|$. Using the notation
\begin{equation*}
\mathit{diag}(S) := \left( \begin{array}{cc} S & 0 \\ 0 & S \end{array} \right)
\end{equation*}
for $S \in \H$ we can express $H_{\lambda, \rho}$ as
\begin{equation*}
H_Z = (Z + J \mathit{diag}(S_{n-1})) \cdots (Z + J\mathit{diag}(S_0)),
\end{equation*}
with $Z$ given by (\ref{d19}).
Hence $\det_{\C}\nolimits H_Z$ is a
complex polynomial of degree $2n$ with zeroes precisely
at $Z=\pm |S_0|$, ..., $\pm |S_{n-1}|$. By Theorem~\ref{dt1} this determinant is invariant
under Darboux transforms. This just corresponds to the fact that the edge lengths
are invariant by construction. Non-trivial further invariants come from the complex
polynomial
\begin{equation*}
P(Z) = \tr_{\C} H_Z
\end{equation*}
of degree $n$. Let us look at the polynomial coefficients of $H_Z$ itself:
\begin{equation*}
    H_Z = \sum_{k=1}^n Z^k A_{n-k},
\end{equation*}
where
\begin{equation*}
    A_k = J^k \sum_{n-1 \ge j_1, ..., j_k \ge 0} \mathit{diag}({S}_{j_1} \cdots
    {S}_{j_k}).
\end{equation*}
In particular,
\begin{eqnarray*}
A_0 & = I, \\
A_1 & = J^k \sum_{k=0}^{n-1} \mathit{diag}({S}_k) = 0,\\
A_2 &= - \sum_{n-1 \ge i > j \ge 0} \mathit{diag}({S}_i {S}_j).
\end{eqnarray*}
That is, $A_2$ is a diagonal matrix with both diagonal entries equal to
\begin{equation*}
q = - \sum_{n-1 \ge i > j \ge 0} S_i S_j.
\end{equation*}
The real part of $q$ is
\begin{eqnarray*}
  \mathrm{Re}(q) & = \sum_{n-1 \ge i > j \ge 0} \pair{S_i, S_j}
           = \frac{1}{2} \sum_{i \neq j} \pair{S_i, S_j}
           = \frac{1}{2} |\sum_{i=0}^{n-1}S_i|^2  
              - \frac{1}{2}\sum_{i=0}^{n-1}|S_i|^2 \\
         & = - \frac{1}{2}\sum_{i=0}^{n-1}|S_i|^2.
\end{eqnarray*}
This is a function of the edge lengths and therefore not interesting.
The imaginary part of $q$ is given by
\begin{eqnarray*}
  2 A := \mathop{\mathrm{Im}}(q) & = - \sum_{i > j} S_i \times S_j \\
  & =  \sum_{j=1}^{n-1} \left(\sum_{i=1}^{n-1}S_i\right) \times S_j \\
  & =  \sum_{j=1}^{n-1} (\gamma_j - \gamma_0) \times (\gamma_{j+1} - \gamma_j)\\
  & =  \sum_{j=1}^{n-1} (\gamma_j - \gamma_0) \times (\gamma_{j+1} -\gamma_0).
\end{eqnarray*}
This invariant $A$ is just the area vector. The following proposition (with
obvious proof) clarifies its geometrical meaning.

\begin{prop} 
  Let $a \in \R^3$ be a unit vector, $|a| = 1$, and endow the plane
  $a^\perp$ with the volume form
  \begin{equation*}
    \det_{a^\perp}\nolimits (X,Y) := \det_{\R^3}\nolimits (a, X, Y).
  \end{equation*}
  Then the area enclosed by the orthogonal projection $\hat{\gamma}$ of
  the polygon $\gamma$
  \begin{equation*}
    \hat{\gamma}_n = \gamma_n - \pair{\gamma_n, a} a
  \end{equation*}
  is equal to $\pair{M, a}$. 
\end{prop}
\noindent
This explains the name \emph{area vector}: It encodes all the
projected areas.
\begin{thm}
\label{areaInvariant}
The area vector $A$ is invariant under Darboux transforms.
\end{thm}
\begin{proof}
By (\ref{d10}), the monodromy matrix of the Darboux transformed curve $\eta$,
\begin{equation*}
H^{\eta}_Z = \sum_{k=0}^n Z^k A^{\eta}_{n-k},
\end{equation*}
satisfies
\begin{equation}
\label{d36}
H^{\eta}_Z (Z + J(-r I + l\,\mathit{diag}{T}_0)) = (Z + J(-r I + l\,\mathit{diag}{T}_0)) H_Z
\end{equation}
Using
\begin{eqnarray*}
H_Z & = & Z^n + Z^{n-2} A_2 + ... + A_0, \\
H^{\eta}_Z & = & Z^n + Z^{n-2} A^{\eta}_2 + ... + A_0^{\eta}
\end{eqnarray*}
and comparing the $Z^{n-2}$-coefficients in both sides of (\ref{d36}) we obtain
\mbox{$A^{\eta}_2 = A_2$}.
\end{proof}

Finally we consider the continuum limit of smooth curves $\gamma: S^1
\rightarrow \R^3$ and indicate why Darboux transforms with small
parameters $l$, $r$ do indeed converge to the smoke ring
flow~(\ref{smokeringflow}). The continuum limit of (\ref{d2}) is
obtained by replacing $S$ by $hS$ and then computing $T' :=
\frac{d}{dh}\big|_{h=0} T_h$. The resulting differential equation is
\begin{equation*}
	T' = (TS-ST)(-r+lT)^{-1}
\end{equation*}
or
\begin{equation}
\label{d38}
 	T' =  \frac{2}{r^2 + l^2} T \times (l T \times S - r S),
\end{equation}
where $S: \R \rightarrow \mathbb{R}^3$ is given by
\begin{equation*}
\gamma' = S.
\end{equation*}
One can check that, as expected, the transformed curve $\eta = \gamma + l T$
satisfies

\begin{equation*}
|\eta'| = |\gamma'|.
\end{equation*}
The monodromy of the ODE (\ref{d38}) is a M\"obius transformation of $S^2$ that
generically has exactly two fixed points. Thus, for generic parameters
$l$ and $r$ a space curve $\gamma$ has exactly two closed Darboux transforms.

Assume now that we have for $r=-l$ a family of such closed Darboux
transforms $\eta_l$ that depend analytically on $l$. Then we
reparametrise $\eta_l$ as
\begin{equation}
\label{eta}
\gamma_l(s) := \eta_l(s-l) = \gamma(s-l)+l T_l(s-l).
\end{equation}
Then $\gamma_0=\gamma$ and comparing coefficients of $l$ in the power series
expansion of (\ref{eta}) we obtain
\begin{equation*}
\frac{\partial}{\partial l}\,\Big|_{l=0}\,\gamma_l = 0, \qquad 
\frac{\partial^2}{\partial l^2}\,\Big|_{l=0}\,\gamma_l = \gamma'\times\gamma''.
\end{equation*}
Hence
\begin{equation*}
\gamma_l-\gamma_0=l^2\gamma'\times\gamma'' + O(l^3).
\end{equation*}
A small time-step $\Delta t$ of the smoke ring flow is therefore approximated by
a Darboux transform with length $l$ given by $l^2=\Delta t$. 

\begin{remark}
In order to eliminate the reparametrising effect of the Darboux transforms
it is convenient to apply first a Darboux transform with 
parameters $l$ and $-r$ followed by a reverse Darboux transform with parameters
$ l$ and $r$. This will cancel out the (first order in $t$) tangential shift and
leave only the (second order in $t$) smoke ring evolution (see \cite{tim}).
\end{remark}

\section{An algorithm for the real time simulation of fluid flow}
\label{simulation}

Based on the theoretic foundations covered in the previous sections, we have 
implemented the following algorithm for the simulation of fluid flow. Our aim 
was to develop an algorithm which is fast enough to generate realistic looking 
computer animations of fluid motion in real time. Figure~\ref{figure3} shows a 
sample screen shot from a simulation which runs smoothly on standard hardware.
\begin{figure}
\centering\includegraphics[width=0.7\textwidth]{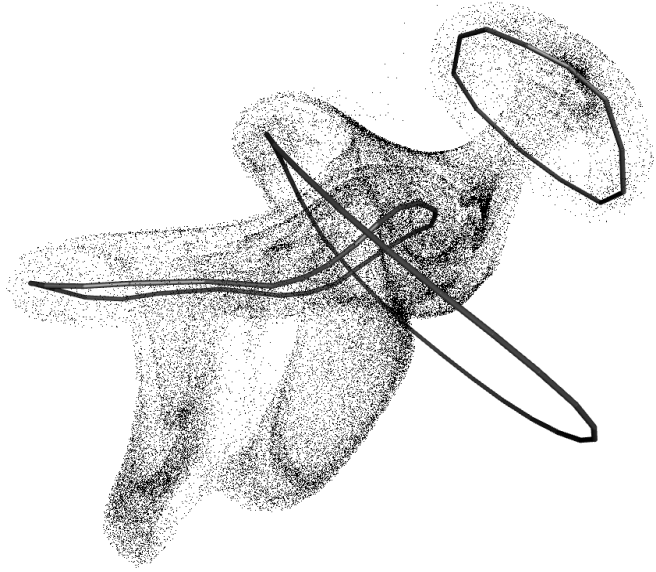} \caption{$256^2$ 
fluid particles evolving under the influence of three polygonal vortex 
filaments.}
\label{figure3}
\end{figure}
We assume the vorticity is concentrated along a few vortex rings,
which we represent by closed polygons. Their motion is governed by a
mixture of the velocity field induced by the polygonal vortex rings
via the smoothed Biot-Savart formula~(\ref{velocityField}) of
Section~\ref{polygonalFilaments}, and Darboux transformations which
approximate a time step of the polygonal smoke ring flow as explained
in Section~\ref{darboux}. The rationale behind this scheme is that the
velocity field induced by an edge of a polygonal vortex filament is
zero on that edge itself. Thus, the adjacent edges do not contribute
to the velocity of a vertex. The Darboux transforms make up for this
lack of local interaction. The following is a summary description of
the algorithm. Details (in particular how we set the parameters $r_i$
and $l_i$ of the Darboux transformation) are given below.

\medskip
\begin{itemize}
\item[] \textbf{input:}
	\begin{itemize}
   \item[$\bullet$] positions $\gamma_{ij}$ of the $j$th vertex of the $i$th 
   polygonal vortex filament $\gamma_i$, where $i=1\ldots m$, $j=1\ldots n_i$.
   \item[$\bullet$] strengths $\Gamma_i$ and smoothing (thickness) parameters 
   $a_i$ of the vortex filaments.
   \item[$\bullet$] positions $p_i\in\R^3$ of advected particles, where 
   $i=1\ldots k$.
   \item[$\bullet$] time-step $\Delta t$ .
    \end{itemize}
   \item[] \textbf{loop:}
		\begin{itemize}
       \item[\small 1] Compute a double Darboux transform $\eta_{i}$ with 
       parameters $\mp r_i, l_i$ of each polygon~$\gamma_i$. 
       $\gamma_{ij} \leftarrow \eta_{ij}$.
       \item[\small 2] Solve 
       $\dot\gamma_{ij}=u(\gamma_{ij})$ for time-step $\Delta t$, where $u(x)$ 
       is the velocity field obtained by the smoothed Biot-Savart 
       formula~(\ref{velocityField}).
       \item[\small 3] Update the particle positions $p_i$ by 
       solving $\dot p_i=u(p_i)$ for time-step~$\Delta t$.
  	\end{itemize}
\end{itemize}

\medskip In Step 1, we determine the parameters $l_i$ and $r_i$ as follows. The 
amount of smoke ring flow needed to make up for the lack of local interaction 
depends on the thickness $a_i$, the number of edges $n_i$ and the total length 
$L_i$ of $\gamma_i$. Since we do not know the correct speed for an arbitrary 
polygon, we determine the parameters for the test case of a regular $n_i$-gon 
with same strength, thickness and total length. We choose the parameters in 
such a way that for the regular $n_i$-gon the sum of self-induced velocity from 
the Biot-Savart formula~(\ref{velocityField}) plus the effect of a double 
Darboux transform coincides with the analytically known speed $U_i$ for a 
circle with same length $L_i$:
\begin{equation}
	\label{eq:U}
U_i = \frac{\Gamma_i}{2 L_i} \left( \ln \frac{4 L_i}{\pi a_i} -1 \right),
\end{equation}
compare~\cite[p.~212]{Saffman}. We compute the self-induced speed $\tilde U_i$ 
of the $n_i$-gon by evaluating the smoothed Biot-Savart 
formula~(\ref{velocityField}) at one vertex for all edges of the $n_i$-gon. 
This speed is slower than $U_i$ because the adjacent edges have no influence on 
a vertex, see Section~\ref{polygonalFilaments}. Now we choose $r_i$ and 
$l_i$ such that a double Darboux transformation translates the regular 
$n_i$-gon by a distance of $(U_i - \tilde U_i)\, \Delta t$. A single Darboux 
transform of the regular $n_i$-gon is a translation in binormal direction plus 
a non-zero rotation about the centre axis. The rotation cancels out for a 
double Darboux transform and is therefore arbitrary. We choose the rotation 
angle to be $2 \pi / n_i$, which leads to the following formulas for $l_i$ and 
$r_i$:
\begin{equation*}
l_i=\sqrt{\left(L_i/n_i\right)^2 + \sigma_i^2}\;, \qquad r_i=\sigma_i 
\cot(\pi/n_i)\, ,
\end{equation*}
where we have abbreviated $\frac 12 (U_i - \tilde U_i)\, \Delta t$ by 
$\sigma_i$.

\medskip In Step~2, we use the fourth order Runge-Kutta scheme (RK4) to solve 
the ordinary the differential equation $\dot x=u(x)$ for the time-step $\Delta 
t$. To advect the large number of particles in Step~3 we use second order 
Runge-Kutta (RK2), where we use the two polygon positions after Step~1 and 
Step~2 as intermediate values. To improve performance further, this step is 
computed on the computer's graphics chip (GPGPU).

Evaluating $u(x)$ via Equation~(\ref{velocityField}) is unproblematic, because 
the integral on the right hand side can be solved explicitly for straight line 
segments; see Equation~(\ref{blurredAntiDerivative}) in 
Section~\ref{polygonalFilaments}.

\bibliographystyle{unsrt}
\bibliography{document}

\end{document}